\let\accentvec\vec
\documentclass[runningheads,a4paper]{llncs}

\let\vec\accentvec

\usepackage{amsmath}
\usepackage{amsfonts}
\usepackage[all]{xy}

\let\hole\undefined

\usepackage{url}
\usepackage{betaneed}

\newcommand{\keywords}[1]{\par\addvspace\baselineskip
\noindent\keywordname\enspace\ignorespaces#1}

\begin{document}

\mainmatter  

\title{The Call-by-need Lambda Calculus, Revisited}

\titlerunning{The Call-by-need Lambda Calculus, Revisited}

%
%
\author{Stephen Chang \and Matthias Felleisen
}
\authorrunning{Stephen Chang \and Matthias Felleisen}

\institute{College of Computer Science\\
           Northeastern University\\
           Boston, Massachusetts, USA\\
\texttt{\{ stchang $\mid$ matthias \} @ ccs.neu.edu}
}

%
%

\toctitle{Lecture Notes in Computer Science}
\tocauthor{Authors' Instructions}
\maketitle

\begin{abstract}
The existing call-by-need \lcs describe lazy evaluation via equational
 logics. A programmer can use these logics to safely ascertain whether one term
 is behaviorally equivalent to another or to determine the value of a lazy
 program. However, neither of the existing calculi models evaluation in a way
 that matches lazy implementations.
 
\hspace{10pt} Both calculi suffer from the same two problems. First, the
 calculi never discard function calls, even after they are completely
 resolved. Second, the calculi include re-association axioms even though these
 axioms are merely administrative steps with no counterpart in any
 implementation.

\hspace{10pt} In this paper, we present an alternative axiomatization of lazy
 evaluation using a single axiom. It eliminates both the function call
 retention problem and the extraneous re-association axioms. Our axiom uses a
 grammar of contexts to describe the exact notion of a \emph{needed
   computation}. Like its predecessors, our new calculus satisfies consistency
 and standardization properties and is thus suitable for reasoning about
 behavioral equivalence. In addition, we establish a correspondence between
 our semantics and Launchbury's natural semantics.

\end{abstract}

\section{A Short History of the $\lambda$ Calculus}

Starting in the late 1950s, programming language researchers began to look
 to Church's \lc~\cite{Church41:CalcOfLambdaConv} for inspiration. Some
 used it as an analytic tool to understand the syntax and semantics of
 programming languages, while others exploited it as the basis for new 
 languages. By 1970, however, a disconnect had emerged in the form of
 call-by-value programming, distinct from the notion of $\beta$ and
 normalization in Church's original calculus. Plotkin~\cite{Plotkin1975LC}
 reconciled the \lc and Landin's SECD machine for the ISWIM language~\cite{pl:iswim} with
 the introduction of a notion of correspondence and with a proof that two
 distinct variants of the \lc corresponded to two distinct variants of the
 ISWIM programming language: one for call-by-value and one for
 call-by-name.

In the early 1970s, researchers proposed  call-by-need~\cite{Friedman1976Cons,Henderson1976Lazy,Wadsworth1971Thesis},
 a third kind of parameter passing mechanism
 that could be viewed as yet another variant of the ISWIM
 language. Call-by-need is supposed to represent the best of both worlds. While
 call-by-value ISWIM always evaluates the argument of a function, the
 call-by-name variant evaluates the argument every time it is needed. Hence, if
 an argument (or some portion) is never needed, call-by-name wins; otherwise
 call-by-value is superior because it avoids re-evaluation of
 arguments. Call-by-need initially proceeds like call-by-name, evaluating
 a function's body before the argument---until the value of the argument is
 needed; at that point, the argument is evaluated and the resulting value is
 used from then onward. In short, call-by-need evaluates an argument at
 most once, and only if needed.

Since then, researchers have explored a number of characterizations of
 call-by-need~\cite{Danvy2010Cbneed,Friedman2007LazyKrivine,Garcia2009Delimited,Josephs1989Lazy,Nakata2009Cbneed,PeytonJones1989Spineless,Purushothaman1992Lazy}.
 Concerning this paper, three stand out. Launchbury's
 semantics~\cite{Launchbury1993Natural} specifies the meaning of complete
 programs with a Kahn-style natural semantics. The call-by-need \lcs of Ariola
 and Felleisen~\cite{Ariola1994Cbneed,Ariola1997Cbneed,Ariola1995Cbneed}, and
 of Maraist, Odersky, and
 Wadler~\cite{Ariola1995Cbneed,Maraist1994Cbneed,Maraist1998Cbneed} are
 equational logics in the spirit of the \lc.

The appeal of the \lc has several reasons. First, a calculus is sound
 with respect to the observational (behavioral) equivalence
 relation~\cite{Morris1968Thesis}. It can therefore serve as the starting
 point for other, more powerful logics. Second, its axioms are rich
 enough to mimic machine evaluation, meaning programmers can reduce
 programs to values without thinking about implementation details. Finally,
 the \lc gives rise to a substantial
 meta-theory~\cite{Barendregt1985LC,CurryFeys1958Combinatory} from which
 researchers have generated useful and practical results for its cousins.

Unfortunately, neither of the existing by-need calculi model lazy evaluation
 in a way that matches lazy language implementations. Both calculi suffer from
 the same two problems. First, unlike the by-name and by-value calculi, the
 by-need calculi never discard function calls, even after the call is resolved
 and the argument is no longer needed. Lazy evaluation does require some
 accumulation of function calls due to the delayed evaluation of arguments but
 the existing calculi adopt the extreme solution of retaining every
 call. Indeed, the creators of the existing calculi acknowledge that a solution
 to this problem would strengthen their work but they could not figure out a
 proper solution.

Second, the calculi include re-association axioms even though these axioms
 have no counterpart in any implementation. The axioms are mere
 administrative steps, needed to construct $\beta$-like redexes. Hence, they
 should not be considered computationally on par with other axioms.

In this paper, we overcome these problems with an alternative
 axiomatization. Based on a single axiom, it avoids the retention of 
 function calls and eliminates the extraneous re-association axioms. The single axiom
 uses a grammar of contexts to describe the exact notion of a \emph{needed
   computation}. Like its predecessors, our new calculus satisfies consistency
 and standardization properties and is thus suitable for reasoning about
 behavioral equivalence. In addition, we establish an intensional
 correspondence with Launchbury's semantics.
 
The second section of this paper recalls the two existing by-need calculi in
 some detail. The third section presents our new calculus, as well as a way to
 derive it from Ariola and Felleisen's calculus. Sections~\ref{sec:essential}
 and \ref{sec:correctness} show that our calculus satisfies the usual
 meta-theorems and that it is correct with respect to Launchbury's
 semantics. Finally, we discuss some possible extensions.

\section{The Original Call-by-need $\lambda$ Calculi}

The original call-by-need \lcs are independently due to two groups: Ariola and
 Felleisen~\cite{Ariola1994Cbneed,Ariola1997Cbneed} and Maraist, et
 al.~\cite{Maraist1994Cbneed,Maraist1998Cbneed}. They were jointly presented at
 POPL in 1995~\cite{Ariola1995Cbneed}. Both calculi use the standard set of
 terms as syntax:
\begin{align*} \tag{Terms} \e = x \mid \lxe \mid \ap{\e}{\e} \end{align*}
 Our treatment of syntax employs the usual
 conventions, including Barendregt's standard hygiene condition
 for variable bindings~\cite{Barendregt1985LC}. Figure~\ref{fig:lneeds}
 specifies the calculus of Maraist et al., \lmow, and \laf,
 Ariola and Felleisen's variant. Nonterminals in some grammar productions have
 subscript tags to differentiate them from similar sets elsewhere in the
 paper. Unsubscripted definitions have the same denotation in all systems.

%
%
\newcommand{\smallmid}{\hspace{-.5pt}\mid\hspace{-0.5pt}}
\begin{figure}[htbp]
\vspace{-10pt}
\hspace{-14pt}\begin{minipage}[b]{0.49\linewidth}
  \begin{align*}
    \valmow &= x \mid \lxe \\
    \ \\
    \C   &= \hole \mid \lx\C \mid \ap{\C}{\e} \mid \ap{\e}{\C} 
  \end{align*}
  \begin{align*}
    \tag{\Vnr} \app{\lxCx}{\valmow} &= \app{\lx\inC{\valmow}}{\valmow} \\
    \tag{\Cnr} \app{\lxe_1}{\ap{\e_2}{\e_3}} &= \app{\lx\ap{\e_1}{\e_3}}{\e_2} \\
    \tag{\Anr} (\lxe_1)(\app{\lye_2&}{\e_3}) = {} \\ &\app{\ly\app{\lxe_1}{\e_2}}{\e_3} \\
    \tag{\Gnr} \app{\lxe_1}{\e_2} &= \e_1,  x \notin \textit{fv}(\e_1)
  \end{align*}
\end{minipage}
\rule[-6pt]{.1pt}{4.5cm}
%
\hspace{-4pt}\begin{minipage}[b]{0.55\linewidth}
  \begin{align*}
    \val &= \lxe \\
    \ansaf &= \val \mid \app{\lx\ansaf}{\e} \\
    \Eaf &= \hole \smallmid \ap{\Eaf}{\e} \smallmid \app{\lx\Eaf}{\e} \smallmid \app{\lxEafx}{\Eaf} 
  \end{align*}
  \begin{align*}
    \tag{\derefnr} \app{\lxEafx}{\val} &= \app{\lx\inEaf{\val}}{\val} \\
    \tag{\liftnr}  \app{\lx\ansaf}{\ap{\e_1}{\e_2}} &= \app{\lx\ap{\ansaf}{\e_2}}{\e_1} \\
    \tag{\assocnr} \appp{\lxEafx}{\app{\ly\ansaf}{\e}} &= \\ &\hspace{-7mm}\app{\ly\app{\lxEafx}{\ansaf}}{\e}\\
    \ 
  \end{align*}
\end{minipage}
\vspace{-10pt}
  \caption{Existing call-by-need \lcs (left: \lmow, right: \laf)}
  \label{fig:lneeds}
\vspace{-10pt}
\end{figure}

In both calculi, the analog to the $\beta$ axiom---also called a
 \textit{basic notion of reduction}~\cite{Barendregt1985LC}---replaces variable
 occurrences, one at a time, with the value 
 of the function's argument. Value substitution means that there is no
 duplication of work as far as argument evaluation is concerned. The function
 call is retained because additional variable occurrences in the function
 body may need the argument.
 Since function calls may accumulate, the calculi come with axioms
 that re-associate bindings to
 pair up functions with their arguments. For example,
 re-associating $\app{\lx\app{\ly\lz z}{\val_y}}{\ap{\val_x}{\val_z}}$ in
 \laf exposes a \derefnr redex:
{
\begin{align*}
\app{\lx\app{\ly\underline{\lz z}}{\val_y}}{\ap{\val_x}{\underline{\val_z\rule{0pt}{6pt}}}} \stackrel{\textit{lift}}{\step} 
\app{\lx\app{\ly\underline{\lz z}}{\ap{\val_y}{\underline{\val_z\rule{0pt}{6pt}}}}}{\val_x} \stackrel{\textit{lift}}{\step} 
\app{\lx\app{\ly\underline{\app{\lz z}{\val_z}}}{\val_y}}{\val_x} 
\end{align*}
}

The two calculi differ from each other in their timing of variable
 replacements. The \lmow calculus allows the replacement of a variable with its
 value anywhere in the body of its binding $\lambda$. The \laf calculus
 replaces a variable with its argument only if evaluation of the function body
 needs it, where ``need'' is formalized via so-called evaluation contexts
 ($\Eaf$). Thus evaluation contexts in \laf serve the double purpose of
 specifying demand for arguments and the standard reduction strategy. The term
 $\app{\lx\ly x}{\val}$ illustrates this difference between the two calculi.
 According to \lmow, the term is a \Vnr redex and reduces to
 $\app{\lx\ly\val}{\val}$, whereas in \laf, the term is irreducible because the
 $x$ occurs in an inner, unapplied $\lambda$, and is thus not ``needed.''

Also, \lmow is more lenient than \laf when it comes to
 re-associations. The \laf calculus re-associates the left or right
 hand side of an application only if it has been completely reduced to an answer,
 but \lmow permits re-association as soon as one nested function layer is
 revealed. In short, \lmow proves more equations than \laf, i.e.,
 $\pmb{\laf} \subset \pmb{\lmow}$.

In \laf, programs reduce to answers:
$\;$ 
$$\evalaf{\e} = \texttt{done} \textrm{ iff there exists an answer } \ansaf
     \textrm{ such that }\pmb{\laf}\vdash\e = \ansaf$$ 
 In contrast, Maraist et al. introduce a ``garbage collection'' axiom into 
 \lmow to avoid answers and to use values instead. This suggests the
 following definition: 
$\;$ 
$$\evalmow{\e} = \texttt{done} \textrm{ iff there exists a value } \valmow
     \textrm{ such that }\pmb{\lmow}\vdash\e =\valmow$$ 
 This turns out to be incorrect, however. Specifically, let $\evalnamename$ be
 the analogous call-by-name evaluator. Then $\evalafname = \evalnamename$ but
 $\evalmowname \neq \evalnamename$. Examples such as ${\app{\lx\ly x}{\Omega}}$ confirm
 the difference.

In recognition of this problem, Maraist et al.\ use Ariola and Felleisen's
 axioms and evaluation contexts to create their Curry-Feys-style standard
 reduction sequences. Doing so reveals the inconsistency of \lmow with 
 respect to Plotkin's \emph{correspondence criteria}~\cite{Plotkin1975LC}. 
 According to Plotkin, a useful calculus \emph{corresponds to} a
 programming language, meaning its axioms (1) satisfy the 
 Church-Rosser and Curry-Feys Standardization properties, and (2) define a 
 standard reduction function that is equal to
 the evaluation function of the programming language. Both the call-by-name and
 the call-by-value \lcs satisfy these criteria with respect to call-by-name and
 call-by-value SECD machines for ISWIM, respectively. So does \laf with respect
 to a call-by-need SECD machine, but some of \lmow's axioms
 cannot be used as standard reduction relations.

Finally, the inclusion of \Gnr is a brute-force attempt to address the
 function call retention problem. Because \Gnr may discard arguments even before
 the function is called, both sets of authors consider it too coarse and
 acknowledge that a tighter solution to the function call retention issue would
 ``strengthen the calculus and its utility for reasoning about the
 implementations of lazy languages''~\cite{Ariola1995Cbneed}.

\section{A New Call-by-need $\lambda$ Calculus}

Our new calculus, \lneed, uses a single axiom, \betaneednr. The new axiom evaluates the argument when it
 is first demanded, replaces all variable occurrences with that result, and
 then discards the argument and thus the function call. In addition, the
 axiom performs the required administrative scope adjustments as part of
 the same step, rendering explicit re-association axioms unnecessary. In short,
 every reduction step in our calculus represents computational progress.

Informally, to perform a reduction, three components must be identified:
\begin{enumerate}
  \item the next demanded variable,
  \item the function that binds that demanded variable,
  \item and the argument to that function.
\end{enumerate}
 In previous by-need calculi the re-association axioms rewrite a term so that
 the binding function and its argument are adjacent.
 
Without the re-association axioms, finding the function that binds the
 demanded variable and its argument requires a different kind of work. The
 following terms show how the demanded variable, its binding function, and its
 argument can appear at seemingly arbitrary locations in a program:
\begin{itemize}
\item $\app{\underline{\lx}\app{\ly\lz\underline{x\rule[-1pt]{0pt}{7pt}}}{\e_y}}{\ap{\underline{\e_x\rule{0pt}{6pt}}}{\e_z}}$
\item $\app{\lx\app{\underline{\ly}\lz\underline{y\rule{0pt}{6pt}}}{\underline{\e_y\rule{0pt}{6pt}}}}{\ap{\e_x}{\e_z}}$
\item $\app{\lx\app{\ly\underline{\lz}\underline{z\rule[-1pt]{0pt}{7pt}}}{\e_y}}{\ap{\e_x}{\underline{\e_z\rule{0pt}{6pt}}}}$
\end{itemize}
 Our \betaneednr axiom employs a grammar of contexts to describe the path
 from a demanded variable to its binding function and from there to its
 argument.

The first subsection explains the syntax and the contexts of \lneed in a
 gradual fashion. The second subsection presents the \betaneednr axiom and also
 shows how to derive it from Ariola and Felleisen's \laf calculus.

\subsection{Contexts}
\label{subsec:newevalcontexts}

Like the existing by-need calculi, the syntax of our calculus is
 that of Church's original calculus. In \lneed, calculations evaluate terms $e$
 to answers $\Av$, which generalize answers from Ariola and Felleisen's
 calculus:
\begin{align*}
  \tag{Terms}
    \e &= x \mid \lxe \mid \ap{\e}{\e} \\
  \tag{Values}
    \val &= \lxe \\
  \tag{Answers}
    \ans &= \Av \\
  \tag{Answer Contexts}
    \A &= \hole \mid \ap{\inA{\lx{\A}}}{\e}
\end{align*}

Following Ariola and Felleisen, the basic axiom uses evaluation contexts to
 specify the notion of demand for variables:
\begin{align*}
   \tag{Evaluation Contexts}
   E = \hole \mid \ap{\E}{\e} \mid \ldots
\end{align*}
 The first two kinds, taken from \laf, specify that a variable is
 demanded, and that a variable in the operator position of an application
 is demanded, respectively.  

Since the calculus is to model program evaluation, we are primarily interested
 in demanded variables under a $\lambda$-abstraction. This kind of evaluation
 context is defined using an answer context $\A$:
 \begin{align*}
   \tag{Another Evaluation Context}
   E = \ldots \mid \inA{\E} \mid \ldots
 \end{align*}
 Using answer contexts, this third evaluation context dictates that demand
 exists under a $\lambda$ if a corresponding argument exists for that
 $\lambda$. Note how \textit{an answer context descends under the same number
   of $\lambda$s as arguments for those $\lambda$s.} In particular, for any term
 $\ap{\inA{\lxe_1}}{\e_2}$, $\e_2$ is always the argument of $\lxe_1$. The
 third evaluation context thus generalizes the function-is-next-to-argument
 requirement found in both call-by-name and call-by-value. The generalization
 is needed due to the retention of function calls in \lneed.

 Here are some example answer contexts that might be used: 
\begin{align*}
\\[-20pt]
  \A_0 &= \underbrace{\app{\lx\hole}{\e_x}} \\
  \A_1 &= \underbrace{\app{\lx\underbrace{\app{\ly\hole}{\e_y}}}{\e_x}} \\
  \A_2 &= \underbrace{\app{\lx\underbrace{\app{\ly\underbrace{\app{\lz\hole}{\e_z}}}{\e_y}}}{\e_x}}
\\[-20pt]
\end{align*}
 An underbrace matches each function to its argument. The examples
 all juxtapose functions and their arguments. In contrast, the next two
 separate functions from their arguments:
\begin{align*}
\\[-20pt]
 \A_3 &= (\underbrace{\lx\ly\lz\hole)\,\e_x}\,\e_y\,\e_z\\[-13pt]
 \phantom{\A_3}&\phantom{=\;(\lx}\underbrace{\phantom{\ly\lz\hole)\,\e_x\,\e_y}}
   \\[-13pt]
 \phantom{\A_3}
   &\phantom{=\;(\lx\ly}\underbrace{\phantom{\lz\hole)\,\e_x\,\e_y\,\e_z}}\\
 \A_4 &= (\lx(\underbrace{\ly\lz\hole)\,\e_y})\,\e_x\,\e_z \\[-13pt]
 \phantom{\A_4}&\phantom{=\;(}\underbrace{\phantom{\lx(\ly\lz z)\,\e_y)\,\e_x}}
   \\[-13pt]
 \phantom{\A_4}
   &\phantom{=\;(\lx(\ly}\underbrace{\phantom{\lz z)\,\e_y)\,\e_x\,\e_z}} 
\\[-20pt]
\end{align*}


To summarize thus far, when a demanded variable is discovered under a
 $\lambda$, the surrounding context looks like this:
\begin{align*}
  \inA{\inE{x}}
\end{align*}
 where both the function binding $x$ and its argument are in $\A$. The
 decomposition of the surrounding context into $\A$ and $\E$ assumes that $\A$
 encompasses as many function-argument pairs as possible; in other words, it is
 impossible to merge the outer part of $\E$ with $\A$ to form a larger answer
 context.

To know which argument corresponds to the demanded variable, we must 
 find the $\lambda$ that binds $x$ in $\A$. To this end, we 
 split answer contexts so that we can ``highlight'' a function-argument
 pair within the context:
 \begin{align*}
   \tag{Partial Answer Contexts--Outer}
     \Ap & = \hole \mid \ap{\inA{\Ap}}{\e} \\
   \tag{Partial Answer Contexts--Inner}
    \Am & = \hole \mid \inA{\lx\Am}
 \end{align*}
 Using these additional contexts, any answer context can be decomposed into
\begin{align*}
 \inAp{\ap{\inA{\lx\inAm{\,\,\,}}}{\e}}
\end{align*}
 where $\e$ is the argument of ${\lx\inAm{\,\,\,}}$.  For a fixed
 function-argument pair in an answer context, this partitioning into $\Ap$,
 $\A$, and $\Am$ is unique. The $\Ap$ subcontext represents the part of the
 answer context around the chosen function-argument pair; the $\Am$ subcontext
 represents the part of the answer context in its body; and $\A$ here is the
 subcontext between the function and its argument. Naturally we must demand
 that $\Ap$ composed with $\Am$ is an answer context as well so that the
 overall context remains an answer context. The following table lists the
 various subcontexts for the example $\A_4$ for various function-argument
 pairs:

$$
 \begin{array}{r|c|c|c}
   & \multicolumn{3}{c}{ \A_4 = \app{\lx\app{\ly\lz\hole}{\e_y}}{\ap{\e_x}{\e_z}} }                              {\rule[-6pt]{0pt}{13pt}} \\ \hline
   \Ap= {}   & \ap{\hole}{\e_z}               & \app{\lx\hole}{\ap{\e_x}{\e_z}} & \hole                          {\rule{0pt}{13pt}}\\
   \A= {}    & \hole                          & \hole                           & \app{\lx\app{\ly\hole}{\e_y}}{\e_x} \\ 
   \Am= {}   & \app{\ly\lz\hole}{\e_y}        & \lz\hole                        & \hole                          {\rule[-6pt]{0pt}{13pt}} \\ \hline
   \A_4= {}  & \inAp{\ap{\inA{\lx\Am}}{\e_x}} & \inAp{\ap{\inA{\ly\Am}}{\e_y}}  & \inAp{\ap{\inA{\lz\Am}}{\e_z}} {\rule{0pt}{13pt}}
\end{array}$$

Now we can define the fourth kind of evaluation context:
\begin{align*}
\tag{Final Eval. Context}
  \E &=  \ldots \;\mid\; \inAp{\ap{\inA{\lx\inAm{\inE{x}}}}{\E}},\;\;
       \textrm{ where } \inAp{\Am} \in \A 
\end{align*}
 This final evaluation context shows how demand shifts to an
 argument when a function parameter is in demand within the function body.

\subsection{The \betaneednr Axiom and a Derivation}
 
Figure~\ref{fig:lneed} summarizes the syntax of \lneed as developed in the
 preceding section.\footnote{We gratefully acknowledge Casey Klein's help
   with the $\A$ production.} In this section we use these
 definitions to formulate the $\beta$ axiom for our calculus.

\begin{figure}[htbp]
\vspace{-22pt}
\begin{align*}
  \tag{Terms}
    \e & = x \mid \lxe \mid \ap{\e}{\e} \\
  \tag{Values}
    \val &= \lxe \\
  \tag{Answers}
    \ans &= \Av \\
  \tag{Answer Contexts}
    \A & = \hole 
      \mid \ap{\inA{\lx{\A}}}{\e}
 \\
  \tag{Partial Answer Contexts--Outer}
    \Ap & = \hole
       \mid \ap{\inA{\Ap}}{\e} 
 \\
  \tag{Partial Answer Contexts--Inner}
    \Am & = \hole
        \mid \inA{\lx\Am} \\
  \tag{Evaluation Contexts} 
    \E & =  \hole 
       \mid \ap{\E}{\e} 
       \mid \inA{\E} 
       \mid \inAp{\ap{\inA{\lx\inAm{\inE{x}}}}{\E}},\\
       &\hspace{3cm}\textrm{ where } \inAp{\Am} \in \A 
\end{align*}
\vspace{-10pt}
  \caption{The syntax and contexts of the new call-by-need \lc, \lneed.} \label{fig:lneed}
\vspace{-10pt}
\end{figure}

Here is the single axiom of \lneed:
\begin{align*}
\\[-20pt]
  \tag{\betaneednr} 
  \inAp{\ap{\inA[_1]{\lx\inAm{\Ex}}}{\Av[_2]}} &= 
  \inAp{\inA[_1]{\inA[_2]{\substx{\inAm{\Ex}}{\val}}}},\\[-3pt]
   &\textrm{where }\inAp{\Am} \in \A
\\[-20pt]
\end{align*}
 A \betaneednr redex determines which parameter $x$ of some function 
 is ``in demand'' and how to locate the corresponding argument $\Av[_2]$,
 which might be an answer not necessarily a value. 
 The contexts from the previous section specify the path from the binding
 position ($\lambda$) to the variable occurrence and the argument. A
 \betaneednr reduction substitutes the value in $\Av[_2]$ for 
 all free occurrences of the function parameter---just like in other
 $\lambda$ calculi. In the process, the function call is discarded. Since the
 argument has been reduced to a value, there is no duplication
 of work, meaning our calculus satisfies the
 requirements of lazy evaluation. Lifting $\A_2$ to the top of the
 evaluation context ensures that its bindings remain intact and visible for $v$.

\def\needstep{\longrightarrow}

Here is a sample reduction in \lneed, where $\longrightarrow$ is the one-step reduction:
\def\innerdemand#1{\underline{#1}}
\begin{eqnarray}
 &           &  \app{\app{\lx\app{\ly\underline{\lz \mathbf{z}\;y\;x}}{\ly y}}{\lx x}}{\underline{\lz z}} \\
 & \needstep & \app{\lx\boxedapp{\ly \app{\lz \innerdemand{z}}{\mathbf{y}\;x}}{\ly y}}{\lx x} \\
 & \needstep & \app{\lx\app{\boxedapp{\lz \mathbf{z}}{\ly y}}{x}}{\lx x}  \\
 & \needstep & \boxedapp{\lx\app{\ly \innerdemand{y}}{\mathbf{x}}}{\lx x} 
\end{eqnarray}
The ``in demand'' variable is in bold; its binding $\lambda$ and argument are
underlined. Line~1 is an example of a reduction that involves a non-adjoined
function and argument pair. In line~2, the demand for the value of
{\it z\/} (twice underlined) triggers a demand for the value of
{\it y\/}; line~4 contains a similar demand chain.

\begin{figure}[htbp]
\vspace{-10pt}
%
\begin{minipage}[b]{0.99\linewidth}
  \begin{align*}
    \tag{Values}
      \val &= \lxe \\
    \tag{Answers}
      \ansaf &= \Aafv \\
    \tag{Answer Contexts}
      \Aaf   &= \hole \mid \app{\lx\Aaf}{\e} \\
    \tag{Evaluation Contexts}
      \Eaf &= \hole \mid \ap{\Eaf}{\e} \mid \inAaf{\Eaf}
                    \mid \app{\lxEafx}{\Eaf}
  \end{align*}
  \begin{align*}
    \tag{\betaneednrr} \app{\lxEafx}{\val}                & =   \substx{\Eafx}{\val} \\
    \tag{\liftnr[']}   \app{\lx\Aafv}{\ap{\e_1}{\e_2}}    & =   \app{\lx\inAaf{\ap{\val}{\e_2}}}{\e_1} \\
    \tag{\assocnr[']}  \appp{\lxEafx}{\app{\ly\Aafv}{\e}} & =   \app{\ly\inAaf{\app{\lxEafx}{\val}}}{\e}
  \end{align*}
\end{minipage}

 \caption{A modified calculus, \lafmod.} \label{fig:lafmod}
\vspace{-5pt}
\end{figure}

To furnish additional intuition into \betaneednr, we use the rest
 of the section to derive it from the axioms of \laf. The \lafmod
 calculus in figure~\ref{fig:lafmod} combines \laf with
 two insights. First, Garcia et al.~\cite{Garcia2009Delimited} observed that when the answers
 in \laf's \liftnr and \assocnr redexes are nested deeply, multiple
 re-associations are performed consecutively. Thus we modify \liftnr and
 \assocnr to perform all these re-associations in one step.\footnote{The same
   modifications cannot be applied to \Cnr and \Anr in \lmow because they allow
   earlier re-association and thus not all the re-associations are performed
   consecutively.} The modified calculus defines answers via answer contexts,
 $\Aaf$, and the modified \liftnr['] and \assocnr['] axioms utilize these
 answer contexts to do the multi-step re-associations. Thus programs in this
 modified calculus reduce to answers $\Aafv$. Also, the $\Aaf$ answer contexts
 are identical to the third kind of evaluation context in \lafmod and the
 new definition of $\Eaf$ reflects this relationship.

Second, Maraist et al.~\cite{Maraist1995Linear} observed that once an argument
 is reduced to a value, all substitutions can be performed at once. The
 \betaneednrr axiom exploits this idea and performs a full
 substitution. Obviously \betaneednrr occasionally performs more substitutions
 than \derefnr. Nevertheless, any term with an answer in \laf likewise has an
 answer when reducing with \betaneednrr.

Next an inspection of the axioms shows that the contractum of a \assocnr['] redex
 contains a \betaneednrr redex. Thus the
 \assocnr['] re-associations and \betaneednrr substitutions can be performed
 with one merged axiom:\footnote{Danvy et al.~\cite{Danvy2010Cbneed}
 dub a \betaneednrr[''] redex a ``potential redex'' in unrelated work.}
\begin{align*}
   \tag{\betaneednrr['']}
   \app{\lxEafx}{\Aafv} &= \inAaf{\substx{\Eafx}{\val}}
\end{align*}

The final step is to merge \liftnr['] with \betaneednrr[''], which requires our
 generalized answer and evaluation contexts. A na{\"\i}ve attempt may look like this:
\begin{align*}
  \tag{\betaneednrr[''']}
    \ap{\inA[_1]{\lxEx}}{\Av[_2]} &= \inA[_1]{\inA[_2]{\substx{\Ex}{\val}}}
\end{align*}
 As the examples in the preceding subsection show, however, the binding
 occurrence for the ``in demand'' parameter $x$ may not be the inner-most
 binding $\lambda$ once the re-association axioms are eliminated. That is, in
 comparison with \betaneednr, \betaneednrr['''] incorrectly assumes $E$ is
 always next to the binder. We solve this final problem with the
 introduction of partial answer contexts.

\section{Consistency, Determinism, and Soundness} \label{sec:essential}

If a calculus is to model a programming language, it must satisfy some
 essential properties, most importantly a Church-Rosser theorem and a
 Curry-Feys standardization theorem~\cite{Plotkin1975LC}. The former guarantees
 consistency of evaluation; that is, we can define an evaluator {\em function\/} with
 the calculus. The latter implies that the calculus comes with a
 {\em deterministic\/} evaluation strategy. Jointly these properties imply
 the calculus is {\em sound\/} with respect to observational equivalence.

\subsection{Consistency: Church-Rosser}

The \lneed calculus defines an evaluator for a by-need language:
$$
\evalneed{\e} = \texttt{done} \textrm{ iff there exists an answer } \ans
\textrm{ such that } \pmb{\lneed} \vdash \e =\ans
$$
 To prove that the evaluator is indeed a (partial) function, we prove that
 the notion of reduction satisfies the Church-Rosser property. 
\begin{theorem}
  $\evalneedname$ is a partial function.
\end{theorem}
\begin{proof}
  The theorem is a direct consequence of lemma~\ref{lem:cr} (Church-Rosser).
\end{proof}

Our strategy is to define a parallel reduction relation for
 \lneed~\cite{Barendregt1985LC}. Define \step to be the compatible closure
 of a \betaneednr reduction, and \steps to be the reflexive, transitive
 closure of \step. Additionally, define \pstep to be the relation that
 reduces \betaneednr redexes in parallel. 

\begin{definition}[$\pstep$]
\begin{center}
$\begin{array}{rll}
  \e                                           & \pstep & \e \\
  \inAp{\ap{\inA[_1]{\lx\inAm{\Ex}}}{\Av[_2]}} & \pstep & \inAp[']{\inA[_1']{\inA[_2']{\substx{\inAm[']{\Ex[']}}{\val'}}}}, \\
  					       &        & \textrm{ if } \inAp{\Am} \in A,\, \inAp[']{\Am'} \in \A,\, \Ap \pstep \Ap', \A_1 \pstep \A_1',\\
					       &        & \hspace{5mm}\A_2 \pstep \A_2', \Am \pstep \Am',\, \E \pstep \E',\, \val \pstep \val'\\
 \ap{\e_1}{\e_2} 			       & \pstep & \ap{\e_1'}{\e_2'}, \textrm{ if } \e_1 \pstep \e_1',\; \e_2 \pstep \e_2' \\
  \lxe					       &\pstep & \lxe', \textrm{ if } \e \pstep \e'\\
\end{array}$
\end{center}
\end{definition}

The parallel reduction relation \pstep relies on notion of parallel reduction for 
 contexts; for simplicity, we overload the relation symbol to denote both
 relations. 

\begin{definition}[$\pstep$ for Contexts]
\begin{center}
$\begin{array}{rll}
  \hole & \pstep & \hole \\
  \ap{\inA[_1]{\lx\A_2}}{\e} & \pstep & \ap{\inA[_1']{\lx\A_2'}}{\e'} , \textrm{ if } \A_1 \pstep \A_1',\, \A_2 \pstep \A_2',\, \e \pstep \e' \\
  \ap{\inA{\Ap}}{\e}         & \pstep & \ap{\inA[']{\Ap'}}{\e'}, \textrm{ if } \A \pstep \A',\, \Ap \pstep \Ap',\, \e \pstep \e' \\
  \inA{\lx\Am}               & \pstep & \inA[']{\lx\Am'}, \textrm{ if } \A \pstep \A',\, \Am \pstep \Am' \\
  \ap{\E}{\e}                & \pstep & \ap{\E'}{\e'}, \textrm{ if } \E \pstep \E',\, \e \pstep \e' \\
  \inA{\E}                   & \pstep & \inA[']{\E'}, \textrm{ if } \A \pstep \A',\, \E \pstep \E' \\
  \inAp{\ap{\inA{\lx\inAm{\inE[_1]{x}}}}{\E_2}} 
  			     & \pstep & \inAp[']{\ap{\inA[']{\lx\inAm[']{\inE[_1']{x}}}}{\E_2'}}, \\
			     &        & \textrm{ if } \inAp{\Am}\in\A,\, \Ap \pstep \Ap',\, \A \pstep \A',\\
			     &        & \hspace{15pt} \Am \pstep \Am',\, \E_1 \pstep \E_1',\, \E_2 \pstep \E_2'
\end{array}$
\end{center}
\end{definition}

\begin{lemma}[Church-Rosser]
\label{lem:cr}
If $\e \steps \e_1$ and $\e \steps \e_2$, then there exists a term $\e'$ such
that $\e_1 \steps \e'$ and $\e_2 \steps \e'$.
\end{lemma}
\begin{proof}
  By lemma~\ref{lem:pstepdiamond}, \pstep satisfies a diamond property. Since
  \pstep extends \step, \steps is also the transitive-reflexive closure of
  \pstep, so \steps also satisfies a diamond property.
\end{proof}
\begin{lemma}[Diamond Property of \pstep]
\label{lem:pstepdiamond}
If $\e \pstep \e_1$ and $\e \pstep \e_2$, there exists $\e'$ such that $\e_1
\pstep \e'$ and $\e_2 \pstep \e'$.
\end{lemma}
\newcommand{\pstepdef}[1]{\textrm{definition } 1.#1}
\newcommand{\diamondcase}[2][1]
           {\case{\e \pstep \e_{#1} \textrm{ by } \pstepdef{#2}}}
\newcommand{\IH}[2][']{\ensuremath{#2 \!\pstep\! #2#1}}
\begin{proof}
  The proof proceeds by structural induction on the derivation of $\e \pstep \e_1$.
\end{proof}
\subsection{Deterministic Behavior: Standard Reduction}

A language calculus should also come with a deterministic algorithm for applying the
 reductions to evaluate a program. Here is our {\em standard reduction\/}:
$$\inE{\e} \srstep \inE{\e'}, \textrm{ where } \e \;\betaneednr\; \e'$$

Our standard reduction strategy picks exactly one redex in a term.

\begin{proposition}[Unique Decomposition]
\label{prop:unique}
  For all closed terms $\e$, $\e$ either is an answer or $\e = E[e']$ for a
  unique evaluation context $\E$ and \betaneednr redex $\e'$.  
\end{proposition}

\begin{proof}
 The proof proceeds by structural induction on $\e$.
\end{proof}

Since our calculus satisfies the unique decomposition property, we can use the 
 standard reduction relation to define a (partial) evaluator function:
$$\evalneedsr{\e} = \texttt{done} \textrm{ iff there exists an answer } 
		   \ans\textrm{ such that }\e\srsteps\ans$$
where $\srsteps$ is the reflexive, transitive closure of $\srstep$.
Proposition \ref{prop:unique} shows \evalneedsrname is a function. The
following theorem confirms that it equals \evalneedname. 
\begin{theorem}
  \label{thm:evalneed=evalneedsr}
  $\evalneedname = \evalneedsrname$
\end{theorem}
\begin{proof}
 The theorem follows from lemma~\ref{lem:sr}, which shows how to obtain a
 standard reduction sequence for any arbitrary reduction sequence. The
 front-end of the former is a series of standard reduction steps. 
\end{proof}
\begin{definition}[Standard Reduction Sequences $\srseqset$] $\;$
\begin{itemize}
\item $x \subset \srseqset$
\item $\srseq{\lxe_1}{\lxe_m}$, if $\srseqe{m}$
\item $\e_0 \srseqop \srseqe{m}$, if $\e_0 \srstep \e_1$ and $\srseqe{m}$
\item $\srseqappend{(\ap{\e_1}{\e_1'})}{(\ap{\e_m}{\e_1'})} \srseq{(\ap{\e_m}{\e_2'})}{(\ap{\e_m}{\e_n'})}$,
  if ${\e_1}\srseqop\cdots\srseqop{\e_m}, \srseqe[']{n}$.
\end{itemize}
\end{definition}
\begin{lemma}[Curry-Feys Standardization] \label{lem:sr}
$\e \steps \e'$ iff there exists $\srseqe{n}$ such that $\e = \e_1$ and $\e' =
  \e_n$.
\end{lemma}
\begin{proof}
 Replace $\steps$ with $\pstep$s, and 
  the lemma immediately follows from lemma~\ref{lem:pstepsr}.
\end{proof}

The key to the remaining proofs is a size metric for parallel reductions.
\begin{definition}[Size of $\pstep$ Reduction]
\begin{center}
$\begin{array}{lcl}
  \size{\e \pstep \e} & = & 0 \\
  \size{\apppp{e_1}{e_2} \pstep \apppp{\e_1'}{\e_2'}} & = & \size{\e_1 \pstep \e_1'} + \size{\e_2 \pstep \e_2'} \\
  \size{\lxe \pstep \lxe'} & = & \size{\e \pstep \e'}\\
  \size{r} & = & 1+ \size{\Ap \pstep \Ap'} + \size{\A_1 \pstep \A_1'} + \size{\inAm{\Ex} \pstep \inAm[']{\Ex[']}} + {}\\
  \multicolumn{3}{r}{ \size{\A_2 \pstep A_2'} + 
		      \numfv{x}{\inAm[']{\Ex[']}} \times \size{\val \pstep \val'}}\\
\multicolumn{3}{r}{\mbox{ where } r = \inAp{\ap{\inA[_1]{\lx\inAm{\Ex}}}{\Av[_2]}} \pstep \inAp[']{\inA[_1']{\inA[_2']{\substx{\inAm[']{\Ex[']}}{\val'}}}} }
\ \\
  \numfv{x}{\e} & = & \mbox{the number of free occurrences of $x$ in $\e$}
\end{array}$
\end{center}
 The size of a parallel reduction of a context equals the sum of the sizes
 of the parallel reductions of the subcontexts and subterms that comprise
 the context.  
\end{definition}

\begin{lemma}
\label{lem:pstepsr}
If $\e_0 \pstep \e_1$ and $\srseq{\e_1}{\e_n}$, there exists
$\e_0 \srseqop {\e_1'} \srseqop \cdots \srseqop {\e_p'} \srseqop \e_n \in \srseqset$.
\end{lemma}

\begin{proof}
  By triple lexicographic induction on (1) length $n$ of the given standard
  reduction sequence, (2) $\size{\e_0 \pstep \e_1}$, and (3) structure of
  $\e_0$.\footnote{We conjecture that the use of Ralph Loader's technique~\cite{Lauder98} may simplify our proof.}
\end{proof}

\subsection{Observational Equivalence}

Following Morris~\cite{Morris1968Thesis} two expressions $\e_1$ and $\e_2$ are
 observationally equivalent, $\e_1 \obseqneed \e_2$, if they are
 indistinguishable in all contexts. Formally, $\e_1 \obseqneed \e_2$ if and
 only if $\evalneed{\inC{\e_1}} = \evalneed{\inC{\e_2}}$ for all contexts $\C$,
 where
\begin{align*}
\tag{Contexts} \C = \hole \mid \lx{C} \mid \ap{\C}{\e} \mid \ap{\e}{\C}
\end{align*}
 An alternative definition of the behavioral equivalence relation uses
 co-induction. In either case, \lneed is sound with respect to observational
 equivalence. 

\begin{theorem}[Soundness] \label{thm:sound}
If $\pmb{\lneed} \vdash \e_1 = \e_2$, then $\e_1 \obseqneed \e_2$.
\end{theorem}
\begin{proof}
 Following Plotkin, a calculus is sound if it satisfies Church-Rosser and Curry-Feys theorems.
\end{proof}

\section{Correctness} \label{sec:correctness}
 
Ariola and Felleisen~\cite{Ariola1997Cbneed} prove that \laf defines
 the same evaluation function as the call-by-name \lc. Nakata and
 Hasegawa~\cite{Nakata2009Cbneed} additionally demonstrate extensional
 correctness of the same calculus with respect to Launchbury's natural
 semantics~\cite{Launchbury1993Natural}. In this section, we show that
 \lneed defines the same evaluation function as Launchbury's 
 semantics. While our theorem statement is extensional, the proof
 illuminates the tight intensional relationship between the two
 systems.

\subsection{Overview}

The gap between the \lneed\ standard reduction ``machine'' and Launchbury's
 natural semantics is huge. While the latter's store-based natural semantics
 uses the equivalent of assignment statements to implement the ``evaluate once,
 only when needed'' policy, the \lneed\ calculus exclusively relies on term
 substitutions. To close the gap, we systematically construct a series of
 intermediate systems that makes comparisons easy, all while ensuring
 correctness at each step. A first step is to convert the natural semantics
 into a store-based machine~\cite{Sestoft1997Deriving}.
 
To further bridge the gap we note that a single-use assignment statement is
 equivalent to a program-wide substitution of shared
 expressions~\cite{Felleisen1989State}. A closely related idea is to reduce
 shared expressions simultaneously. This leads to a parallel program rewriting
 system, dubbed \lstep. Equipped with \lstep we get closer to \lneed\ but not
 all the way there because reductions in \lneed and \lstep are too
 coarse-grained for direct comparison. Fortunately, it is easy to construct an
 intermediate transition system that eliminates the remainder of the gap. We
 convert \lneed to an equivalent CK transition
 system~\cite{Felleisen2009Redex}, where the program is partitioned into a
 control string (C) and an explicit context (K) and we show that there is a
 correspondence between this transition system and \lstep.

\begin{figure}[htbp]
  $$\xymatrixcolsep{13mm}\xymatrixrowsep{6mm}\xymatrix{
    \textrm{store machine} \hspace{-12mm}
    \ar@<12mm>[d]_*+{\buildLname}
    &
    \ar@{|->}[r] & 
    \ar@{|->}[r]^{:=} \ar@{-->}[dl] & 
    \ar@{|->}[r] \ar@{-->}[drr] & 
    \ar@{|->}[r] & \\
    \lstep \hspace{-25mm} 
    &
    \ar@{|->}[rrrr]^{\betastep} & & & & \\
    \textrm{CK transitions} \hspace{-12mm} 
    \ar@<-12mm>[u]^*+{\buildtostepname}
    \ar@<12mm>[d]_*+{\buildname}
    &
    \ar@{|->}[r] & 
    \ar@{|->}[r]^{go\;under\;\lambda} \ar@{-->}[ul]   & 
    \ar@{|->}[r] \ar@{-->}[urr]  & 
    \ar@{|->}[r]^{\betaneedck} \ar@{-->}[dlll] & \ar@{-->}[d] \\
    \lneed \hspace{-25mm}  &
    \ar@{|->}[rrrr]^{\betaneednr} & & & & \\
  }$$
  \caption{Summary of correctness proof technique.}
  \label{fig:correctness}
\end{figure}
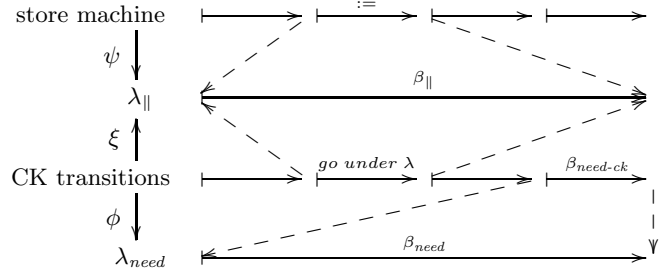

Figure~\ref{fig:correctness} outlines our proof strategy pictorially. The four
 horizontal layers correspond to the four rewriting systems. While \lneed and
 \lstep use large steps to progress from term to term, the machine-like systems
 take several small steps. The solid vertical arrows between the layers figure
 indicate how mapping functions relate the rewriting sequences and the dashed
 arrows show how the smaller machine steps correspond to the larger steps of
 \lneed and \lstep:
\begin{itemize}
  \item The \buildLname function maps states from the store-based machine to
    terms in the \lstep world. For every step in the natural-semantics machine,
    the resulting operation in \lstep is either a no-op or a \betastep
    reduction, with assignment in the store-machine being equivalent to
    program-wide substitution in \lstep.
  \item Similarly, the \buildtostepname function maps states of the CK
    transition system to the \lstep space and for every CK transition, the
    resulting \lstep operation is also either a no-op or a \betastep reduction,
    with the transition that descends under a $\lambda$-abstraction being
    equivalent to substitution in \lstep.
  \item Finally, the \buildname function maps states of the CK transition system to
    \lneed terms and is used to show that the CK system and \lneed are
    equivalent. For every CK transition, the equivalent \lneed operation is
    either a no-op or a \betaneednr reduction.
\end{itemize}

\begin{figure*}[htbp]

$\begin{array}{l@{\qquad\qquad\qquad\qquad\qquad\qquad}rcl@{\qquad\qquad\qquad\qquad}r}
 \multicolumn{2}{l}{\underline{Syntax}} \\
    & \SL    & = & \ckh{\e}{\FLs}{\Gamma}   & \mbox{(States)} \\
    & \FLs   & = & \FLdots                  & \mbox{(List of Frames)}\\
    & \FL    & = & \argFe \mid \varFx       & \mbox{(Frames)} \\
    & \Gamma & = & (\heapelt{x}{\e},\ldots) & \mbox{(Heaps)}
  \end{array}$

$$\begin{array}{lrcl@{\ }r}
 \multicolumn{2}{l}{\underline{Transitions}}\\
    & \ckh{\ap{\e_1}{\e_2}}{\FLs}{\Gamma}    &\ckhstep & \ckhp{\e_1}{\argFe[_2],\FLs}{\Gamma}                     & \mbox{(\pushargckh)} \\
    & \ckhp{\lxe_1}{\argFe[_2],\FLs}{\Gamma} &\ckhstep & \ckh{\substx{\e_1}{y}}{\FLs}{(\Gamma,\heapelt{y}{\e_2})}, 
      					     	       \ \fresh{y}                                            & \mbox{(\descendlamckh)}\\
    & \ckh{x}{\FLs}{(\Gamma,x\mapsto\e)}     &\ckhstep & \ckhp{\e}{\varFx,\FLs}{\Gamma}                           & \mbox{(\lookupvarckh)}\\
    & \ckhp{\val}{\varFx,\FLs}{\Gamma}	     &\ckhstep & \ckh{\val}{\FLs}{(\Gamma,x\mapsto\val)}                  & \mbox{(\updateheapckh)}
  \end{array}$$

  \caption{The natural semantics as an abstract machine.} \label{fig:launchburymachine}
\end{figure*}

Subsections~\ref{subsec:launchbury} and \ref{subsec:lstep} present the
 store-based machine and the parallel rewriting semantics, respectively,
 including a proof of equivalence. Subsection~\ref{subsec:ck} presents the CK
 system and subsection~\ref{subsec:relatingmachines} explains the rest of the
 proof.

\subsection{Adapting Launchbury's Natural Semantics}
\label{subsec:launchbury}

Figure~\ref{fig:launchburymachine} describes the syntax and
 transitions of the store machine.\footnote{To aid comparisons, we slightly
   alter Launchbury's rules (and the resulting machine) to use pure $\lambda$
   terms. Thus we avoid Launchbury's preprocessing and special syntax.}  It is
 dubbed CKH because it resembles a three-register
 machine~\cite{Felleisen2009Redex}: a machine state $\SL$ is comprised of a
 control string (C), a list of frames (K) that represents the control context
 in an inside-out manner, and a heap (H). The $\ldots$ notation means ``zero of
 more of the preceding kind of element.'' An $\argFe$ frame represents the
 argument in an application and the $\varFx$ frame indicates that a heap
 expression is the current control string. Parentheses are used to group a list
 of frames when necessary. The initial machine state for a program $\e$ is
 $\ckhe{()}{()}$. Computation terminates when the control string is a value and the list
 of frames is empty.

The \pushargckh transition moves the argument in an application to a new
 \texttt{arg} frame in the frame list and makes the operator the next control
 string. When that operator is a $\lambda$-abstraction, the \descendlamckh
 transition adds its argument to the heap,\footnote{The notation
   $(\Gamma,\heapelt{x}{\e})$ is a heap $\Gamma$, extended with the
 variable-term mapping $\heapelt{x}{\e}$.} mapped to a fresh variable name,
 and makes the body of the operator the new control string. The \lookupvarckh
 transition evaluates an argument from the heap when the control string is a
 variable. The mapping is removed from the heap and a new $\varFx$ frame
 remembers the variable whose corresponding expression is under
 evaluation. Finally, when the heap expression is reduced to a value, the
 \updateheapckh transition extends the heap again.

\subsection{Parallel Rewriting}
\label{subsec:lstep}

The syntax of the parallel $\lambda$-rewriting semantics is as follows:
\begin{align*}
  \tag{Terms}
  \estep &= \e \mid \labxe \\
  \tag{Values}
  \vstep &= \val \mid \labx{\vstep}\\
  \tag{Evaluation Contexts}
  \Estep &= \hole \mid \ap{\Estep}{\estep} \mid \labx{\Estep}
\end{align*}
 This system expresses computation with a selective parallel reduction
 strategy. When a function application is in demand, the system substitutes
 the argument for all free occurrences of the bound variable, regardless of
 the status of the argument. When an instance of a substituted argument is
 reduced, however, all instances of the argument are reduced in
 parallel. Here is a sample reduction:

\begin{align*}
  \appp{\lx\ap{x}{x}}{\ap{I}{I}} \stepstep 
  \ap{\labx{(\ap{I}{I})}}{\labx{(\ap{I}{I})}} \stepstep 
  \ap{\labx{I}}{\labx{I}} \stepstep I
\end{align*}
%

The \lstep semantics keeps track of arguments via labeled
 terms $\labxe$, where labels are variables. Values in \lstep also include
 labeled $\lambda$-abstractions. Reducing a labeled term triggers the
 simultaneous reduction of all other terms with the same
 label. Otherwise, labels do not affect program evaluation.

We require that all expressions with the same label must be identical.
\begin{definition}
  A program $\estep$ is \emph{consistently labeled (CL)} when for any two subterms $\labx[_1]{\estep_1}$ and
  $\labx[_2]{\estep_2}$ of $\estep$, $x_1 = x_2$ implies $\estep_1 = \estep_2$.
\end{definition}

In the reduction of \lstep programs, evaluation contexts $\Estep$ determine
 which part of the program to reduce next. The \lstep evaluation contexts are
 the call-by-name evaluation contexts with the addition of the labeled
 $\labx{\Estep}$ context, which dictates that a redex search goes under labeled
 terms. Essentially, when searching for a redex, terms tagged with a label are
 treated as if they were unlabeled.
 
The parallel semantics can exploit simpler evaluation contexts than \lneed
 because substitution occurs as soon as an application is encountered:
$$
 \begin{array}{ll}
   \inEstep{\app{\labvec{y}{(\lx\estep_1)}}{\estep_2}}\stepstep
   \begin{cases}
    \inEstep{\estep},     \hspace{1.1cm} \textrm{if \hole\ is not under a label in } \Estep \\
    \substlab{ \inEstep{\estep} }{ z }{ \inEstep[_2]{\estep} }, \textrm{ if } \inEstephole = \inEstep[_1]{\labz{(\inEstephole[_2])}}\\
       \multicolumn{2}{r}{\hspace{1.6cm}\textrm{and \hole\ is not under a label in } \Estep_2 }
   \end{cases} &								\hspace{12pt} (\betastep)\\
   \multicolumn{2}{c}{\textrm{where } \estep = \substx{\estep_1}{\lab{w}{\estep_2}}, w \textrm{ fresh}}
 \end{array}
$$
 On the left-hand side of \betastep, the program is partitioned into a context
 and a $\beta$-like redex. A term 
 $\labvec{y}{\e}$ may have any number of labels and possibly
 none. On the right-hand side, the redex is contracted to a term
 $\substx{\estep_1}{\lab{w}{\estep_2}}$ such that the argument is tagged with
 an unique label $w$. Obsolete labels $\vec{y}$ are discarded.
%
%
%
%

There are two distinct ways to contract a redex: when the redex is not under
 any labels and when the redex occurs under at least one label. For the former,
 the redex is the only contracted part of the program. For the latter,  all other instances of that labeled term are
 similarly contracted. In this second case, the evaluation context is further
 subdivided as $\inEstephole = \inEstep[_1]{\labz{(\inEstephole[_2])}}$, where
 $z$ is the label nearest the redex, i.e., $\Estep_2$ contains no additional
 labels. A whole-program substitution function is used to perform the parallel
 reduction:
  \begin{align*}
    \substlabxe{\labx{\estep_1}} &= \labxe \\
    \substlabye{\labx{\estep_1}} &= 
    \labx{(\substlabye{\estep_1})},
        \; x \neq y  \\
    \substlabxe{(\lx\estep_1)} &= \lx(\substlabxe{\estep_1}) \\
    \substlabxe{\apppp{\estep_1}{\estep_2}} &=
    \apppp{\substlabxe{\estep_1}}
          {\substlabxe{\estep_2}}\\
    \textrm{otherwise, }
    \substlabxe{\estep_1} &= \estep_1
  \end{align*}
%
%
%
%

Rewriting terms with $\betastep$ preserves the consistent labeling property.
\begin{proposition}
If $\estep$ is CL and $\estep \stepstep \estep '$, then $\estep'$ is CL.
\end{proposition}

The \buildLname function reconstructs a \lstep term from a CKH machine
 configuration:
  \begin{align*}
    %
    %
    \tag*{\boxed{\buildLname : \SL \rightarrow \estep}} \\[-20pt]
    \buildL{\ckhe{(\varFx,\FLs)}{\Gamma}} &= 
      \buildL{\ckh{x}{\FLs}{(\heapelt{x}{\e},\Gamma)}} \\
    \buildL{\ckhe[_1]{(\argFe[_2],\FLs)}{\Gamma}} &=
      \buildL{\ckh{\ap{\e_1}{\e_2}}{\FLs}{\Gamma}} \\
    \buildL{\ckhe{()}{\Gamma}} &= \substheap{\e}{\Gamma}
  \end{align*}
 The operation $\substheap{\e}{\Gamma}$, using overloaded notation, replaces all free
 variables in $\e$ with their corresponding terms in $\Gamma$ and tags them with
 appropriate labels.
%
%

Lemma~\ref{lem:ckh->lstep} demonstrates the bulk of the equivalence of the
 store machine and \lstep.\footnote{The lemma relies on an extension of the
   typical $\alpha$-equivalence classes of terms to include variables in labels
   as well.} The rest of the equivalence proof is
 straightforward~\cite{Felleisen2009Redex}.

\begin{lemma}
\label{lem:ckh->lstep}
  If $\ckhe{\FLs}{\Gamma} \ckhstep \ckhe[']{\FLs'}{\Gamma'}$, then either:
  \begin{enumerate}
  \item $\buildL{\ckhe{\FLs}{\Gamma}} = \buildL{\ckhe[']{\FLs'}{\Gamma'}}$
  \item $\buildL{\ckhe{\FLs}{\Gamma}}\stepstep\buildL{\ckhe[']{\FLs'}{\Gamma'}}$
  \end{enumerate}
\end{lemma}
%
%

\subsection{A Transition System for Comparing \lneed and \lstep} \label{subsec:ck}

The CK layer in figure~\ref{fig:correctness} mediates between \lstep and
 \lneed. The corresponding transition system resembles a two-register CK 
 machine~\cite{Felleisen2009Redex}. Figure~\ref{fig:ck}
 describes the syntax and the transitions of the system.\footnote{The CK
 transition system is a proof-technical device. Unlike the original CK
 machine, ours is ill-suited for an implementation.}
\begin{figure*}[htbp]
%
%
$\begin{array}{l@{\qquad\qquad\qquad\qquad\qquad\qquad}rcl@{\qquad\quad}r}
\multicolumn{2}{l}{\underline{Syntax}} \\
    & \St & = & \cke{\Fs}			        & \mbox{(States)}\\
    & \Fs & = & \Fdots			                & \mbox{(List of Frames)}\\
    & \F  & = & \argFe \mid \lxF \mid \bodyFx{\Fs}{\Fs} & \mbox{(Frames)}
\end{array}$

$$\begin{array}{lrcl@{\quad}r}
\multicolumn{2}{l}{\underline{Transitions}}\\
  & \ckFs{\ap{\e_1}{\e_2}}                      & \ckstep & \cke[_1]{\consFs{\argFe[_2]}}                    & \mbox{(\pushargck)} \\
  & \ckFs{\lxe}                                 & \ckstep & \cke{\consFs{\lxF}}				     & \mbox{(\descendlamck)}\\
  & 						&         & \mbox{ if }\bal{\Fs} > 0 \\
  & \ckp{x}{\Fs_1,\lxF,\Fs_2,\argFe,\Fs}        & \ckstep &\ckep{\bodyFx{\Fs_1}{\Fs_2},\Fs}		     & \mbox{(\lookupvarck)} \\
    \multicolumn{4}{r}{\mbox{if } \buildF{\Fs_1} \in \inAm{\E},
    				  \buildF{\Fs_2} \in \A, 
				  \buildF{\Fs} \in \inE{\Ap},
				  \inAp{\Am} \in A}\\
  & \ckp{\val}{\Fs_3,\bodyFx{\Fs_1}{\Fs_2},\Fs} & \ckstep & \ckp{\val}{\substx{\Fs_1}{\val},\Fs_3,\Fs_2,\Fs} & \mbox{(\betaneedck)}\\
  & 						& 	  & \mbox{ if } \buildF{\Fs_3} \in \A 
\end{array}$$

\caption{A transition system for comparing \lneed and \lstep.}  \label{fig:ck}
\end{figure*}

States consist of a subterm and a list of frames representing the context. The
 first kind of frame represents the argument in an application and the second
 frame represents a $\lambda$-abstraction with a hole in the body. The last
 kind of frame has two frame list components, the first representing a context
 in the body of the $\lambda$, and the second representing the context between
 the $\lambda$ and its argument. The variable in this last frame is the
 variable bound by the $\lambda$ expression under evaluation. The initial state
 for a program $\e$ is $\ckemt$, where () is an empty list of frames, and
 evaluation terminates when the control string is a value and the list of
 frames is equivalent to an answer context.

The \pushargck transition makes the operator in an application the new control
 string and adds a new \texttt{arg} frame to the frame list containing the
 argument. The \descendlamck transition goes under a $\lambda$, making the body
 the control string, but only if that $\lambda$ has a corresponding argument in
 the frame list, as determined by the \balancename function, defined as follows:
$$\begin{array}{rcl}
  \multicolumn{3}{r}{ {\boxed{\balancename : \Fs \rightarrow \mathbb{Z}}} }\\
 \bal{\Fs_3,\bodyFx{\Fs_1}{\Fs_2},\Fs} & = & \bal{\Fs_3} \\
                             \bal{\Fs} & = & \texttt{\#arg-frames}(\Fs) - \texttt{\#lam-frames}(\Fs) \\
			               &   & \quad \Fs \textrm{ contains no } \texttt{bod}\textrm{ frames}
\end{array}$$
 The \balancename side condition for \descendlamck dictates that evaluation
 goes under a $\lambda$ only if there is a matching argument for it, thus
 complying with the analogous evaluation context. The \balancename function
 employs \texttt{\#arg-frames} and \texttt{\#lam-frames} to count the number of \texttt{arg} or
 \texttt{lam} frames, respectively, in a list of frames. Their definitions are
 elementary and therefore omitted.

The \lookupvarck transition is invoked if the control string is a variable,
 somewhere in a $\lambda$ body, and the rest of the frames have a certain shape
 consistent with the corresponding parts of a \betaneednr redex. With this
 transition, the argument associated with the variable becomes the next control
 string and the context around the variable in the $\lambda$ body and the
 context between the $\lambda$ and argument are saved in a new \texttt{bod}
 frame. Finally, when an argument is an answer, indicated by a value control
 string and a \texttt{bod} frame in the frame list---with the equivalent of an
 answer context in between---the value gets substituted into the body of the
 $\lambda$ according to the \betaneedck transition. The \betaneedck transition
 uses a substitution function on frame lists, $\substx{\Fs}{\e}$, which
 overloads the notation for regular term substitution and has the expected
 definition.

%
%


%
\begin{figure}[tbhp]
%
%
%
\begin{minipage}[b]{0.44\linewidth}
  \begin{align*}
    \tag*{\boxed{\buildname : \St \rightarrow \e}}\\
    \build{\ckeFs} &= \inhole{\buildF{\Fs}}{\e} \\
    \\[-7pt]
    \tag*{\boxed{\buildFname : \Fs \rightarrow \E}}\\
    \buildF{()} &= \hole \\
    \buildF{\lxF,\Fs} &= \inhole{\buildF{\Fs}}{\lx\hole} \\
    \buildF{\argFe,\Fs} &= \inhole{\buildF{\Fs}}{\ap{\hole}{\e}} \\
    \buildF{\bodyFx{\Fs_1}{\Fs_2},\Fs} &=\\
    &\hspace{-2.2cm}\inhole{\buildF{\Fs}}
                  {\ap{\inhole{\buildF{\Fs_2}}
                             {\lx\inhole{\buildF{\Fs_1}}
                                        {x}}}
                     {\hole}}\\[14pt]
\end{align*}
\end{minipage}
\rule{0.5pt}{5.8cm}
\begin{minipage}[b]{0.54\linewidth}
  \begin{align*}
    %
    %
    %
    \tag*{\boxed{\buildtostepname : \St \rightarrow \estep }}\\
    \buildtostep{\ckeFs} &= \buildtosteppnop{\Fs}{\e}\\
    \\[-7pt]
    \tag*{\boxed{\buildtosteppname : \Fs \times \estep \rightarrow \estep }}\\
    \buildtosteppe{\;} &= \estep \\
    \buildtostepp{\argF{\estep_1},\Fs}{\estep}
        &=
    \buildtosteppnop{\Fs}{\ap{\estep}{\estep_1}} \\
    \buildtostepp{\bodyFx{\Fs_1}{\Fs_2},\Fs}{\estep} &=\\[-1pt]
    &\hspace{-19mm}\buildtostepp{\Fs_1,\lxF,\Fs_2,\argF{\estep},\Fs}{x}\\[-29pt]
  \end{align*}
  \begin{align*}
  \buildtosteppe{\lxF,\Fs_1,\argF{\estep_1},\Fs_2}  =\hspace{15mm}\\
    \buildtostepp{\Fs_1,\Fs_2}{\substx{\estep}{\laby{\estep_1}}}\\[-2pt]
      \buildF{\Fs_1} \in \A,\,\,y \textrm{ fresh}\\[-12pt]
  \end{align*}
\end{minipage}
\vspace{-15pt}
  \caption{Functions to map CK states to \lneed (\buildname) and \lstep (\buildtostepname).}
  \label{fig:CKtolneedandlstep}
\vspace{-10pt}
\end{figure}

Figure~\ref{fig:CKtolneedandlstep} defines metafunctions for the CK transition
 system. The \buildname function converts a CK state to the equivalent \lneed
 term, and uses \buildFname to convert a list of frames to an evaluation
 context.

Now we can show that an evaluator defined with $\,\ckstep\,$ is equivalent to
 $\evalneedsrname$. The essence of the proof is a lemma that relates the shape
 of CK transition sequences to the shape of \lneed standard reduction
 sequences. The rest of the equivalence proof is
 straightforward~\cite{Felleisen2009Redex}.
\begin{lemma}
\label{lem:ck->lneed}
  If $\ckeFs \ckstep \cke[']{\Fs'}$, then either:
  \begin{enumerate}
   \item $\build{\ckeFs} = \build{\cke[']{\Fs'}}$
   \item $\build{\ckeFs} \srstep \build{\cke[']{\Fs'}}$
  \end{enumerate}
\end{lemma}

%
%

Finally, we show how the CK system corresponds to \lstep. The \buildtostepname
 function defined in figure~\ref{fig:CKtolneedandlstep} constructs a \lstep
 term from a CK configuration.
\begin{lemma}
\label{lem:ck->lstep}
  If $\ckeFs \ckstep \cke[']{\Fs'}$, then either:
  \begin{enumerate}
  \item $\buildtostep{\ckeFs} = \buildtostep{\cke[']{\Fs'}}$
  \item $\buildtostep{\ckeFs} \stepstep \buildtostep{\cke[']{\Fs'}}$
  \end{enumerate}
\end{lemma}
%
%


\subsection{Relating all Layers} \label{subsec:relatingmachines}

In the previous subsections, we have demonstrated the correspondence between
 \lstep, the natural semantics, and the \lneed standard reduction sequences via
 lemmas~\ref{lem:ckh->lstep} through \ref{lem:ck->lstep}. We conclude this
 section with the statement of an extensional correctness theorem, where
 \evalnsname is an evaluator defined with the store machine transitions. The
 theorem follows from the composition of the equivalences of our specified
 rewriting systems.

\begin{theorem}  $\evalneedname = \evalnsname$ \end{theorem}

\section{Extensions and Variants}

\textbf{Data Constructors} Real-world lazy languages come with data structure
 construction and extraction operators. Like function arguments, the arguments
 to a data constructor should not be evaluated until there is demand for their
 values~\cite{Friedman1976Cons,Henderson1976Lazy}. The standard \lc encoding of
 such operators~\cite{Barendregt1985LC} works well:
\begin{align*}
  \texttt{cons} = \lx\ly\lm{s}\apthree{s}{x}{y}, \quad
  \texttt{car}  = \lm{p}\ap{p}{\lx\ly x}, \quad
  \texttt{cdr}  = \lm{p}\ap{p}{\lx\ly y}
\end{align*}
Adding true algebraic systems should also be straightforward.

\bigskip

\noindent\textbf{Recursion} Our \lneed calculus represents just a core \lc and
 does not include an explicit \texttt{letrec} constructor for cyclic
 terms. Since cyclic programming is an important idiom in lazy programming
 languages, others have extensively explored cyclic by-need calculi,
 e.g., Ariola and Blum~\cite{Ariola1997Cyclic}, and applying their solutions to our
 calculus should pose no problems.

\section{Conclusion}
 
Following Plotkin's work on call-by-name and call-by-value, we present a
 call-by-need \lc that expresses computation via a single axiom in the spirit
 of $\beta$. Our calculus is close to implementations of lazy languages because
 it captures the idea of by-need computation without retaining every function
 call and without need for re-associating terms. We show that our calculus
 satisfies Plotkin's criteria, including an intensional correspondence
 between our calculus and a Launchbury-style natural semantics. Our future work
 will leverage our \lneed calculus to derive a new abstract machine for lazy
 languages.

\paragraph{Acknowledgments} We thank J. Ian Johnson, Casey Klein, Vincent 
 St-Amour, Asumu Takikawa, Aaron Turon, Mitchell Wand, and the ESOP
 2012 reviewers for their feedback on early drafts. This work was supported in
 part by NSF Infrastructure grant CNS-0855140 and AFOSR grant FA9550-09-1-0110.

\bibliographystyle{splncs03}
\bibliography{betaneed}

\end{document}